\documentclass[pdftex,letterpaper,12pt]{article}{\twocolumn}
\usepackage[top=1in,bottom=1in,left=1in,right=1in]{geometry}	
\usepackage[dvips]{graphicx}
\usepackage{amsmath, amssymb, amsthm, bbm, mathtools, commath,amsfonts}
\usepackage{color}
\usepackage{datetime}
\usepackage{fancyhdr}
\usepackage{fancyvrb}
\usepackage{hyperref}
\usepackage{pslatex}
\usepackage{verbatim}
\usepackage{lmodern}
\usepackage{tikz-cd}
\usepackage{cleveref}
\usepackage{dsfont}
\usepackage{enumitem}
\usepackage{helvet}

\author{Alex Kokot, Charles Johnson}

\setlist[enumerate,1]{label={(\arabic*)}}

\newcommand{\CC}{\mathbb C}

\newcommand{\RR}{\mathbb R}

\DeclarePairedDelimiterX{\inp}[2]{\langle}{\rangle}{#1, #2} 

\newtheorem{theorem}{Theorem}[section]
\newtheorem{lemma}[theorem]{Lemma}
\newtheorem{corollary}[theorem]{Corollary}
\theoremstyle{definition}

\theoremstyle{definition}

\theoremstyle{definition}

\begin{document}

\title{Characteristics of Eigenvalues Realized by Path-Connected Sets of Matrices}

\author{Alex Kokot, Charles Johnson}

\date{ }
\maketitle

\section{Introduction}

It is a well known fact \cite{Kato} that a continuous path in the space of $n\times n$ complex matrices $A(x):\mathbb{R}\to M_n$ induces continuous eigenvalue paths in the parameter $x$. In particular, we can choose $n$ of these paths so that each of them corresponds to an eigenvalue of $A(x)$. 

\begin{figure}[hbt!]
    \centering
    \includegraphics[scale=.2]{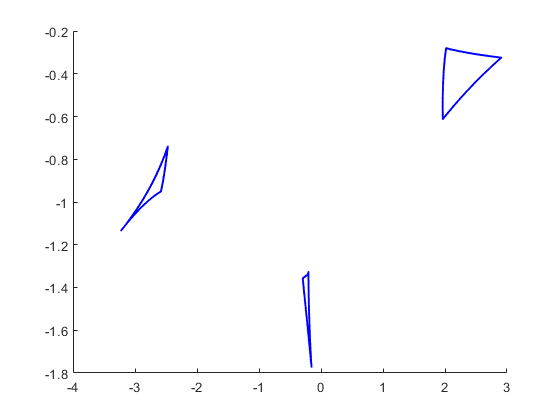}
    \includegraphics[scale=.2]{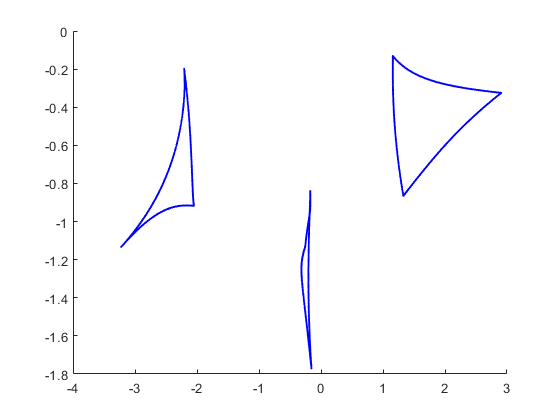}
    \includegraphics[scale=.2]{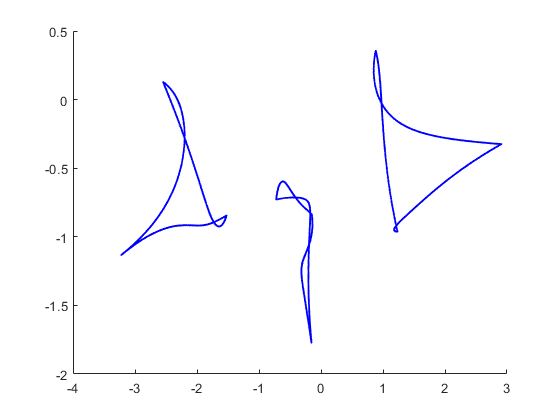}
    \includegraphics[scale=.2]{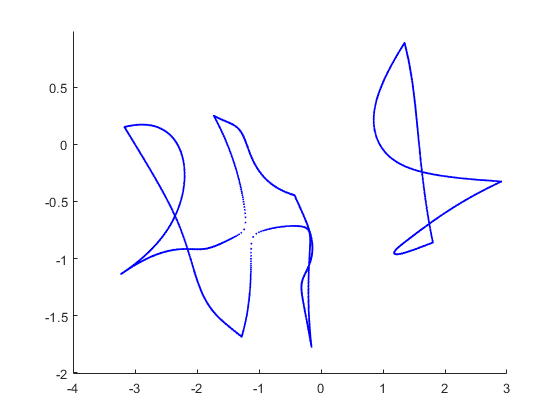}
    \caption{Th eigenvalues realized by three different matrix paths.}
\end{figure}

These paths will be our primary tool to study the eigen-surface associated with a set of matrices $S$, denoted $ES(S)\subseteq \mathbb{C}^{k+1}$. If 
$$
S \subset \operatorname{span}\{A_1,A_2,\dots,A_k\}
$$
where the matrices $A_i$ are linearly independent, and $S$ is not contained in the span of any sub-collection of these matrices, then we can define\footnote{Here we use $\sigma(A)$ to refer to the set of eigenvalues of $A\in M_n$, that is, the scalars $\lambda$ such that $\det(A-\lambda I) = 0$.} the Eigen-Surface of $S$ to be the set
\begin{gather*}
   ES(S):=\{(\alpha_1,\alpha_2,\dots,\alpha_k, \lambda)\ |\ A\in S,\\ \lambda\in \sigma(A),\ A= \alpha_1 A + \alpha_2 A_2 + \dots + \alpha_k A_k \}.
\end{gather*}

We define $\sigma^*(A):= (\alpha_1,\dots,\alpha_k)\times \sigma(A)$ to distinguish elements in this higher dimensional space, and we will refer to its elements as $\lambda(A)\in \sigma^*(A)$. For eigenpaths we write $\lambda(A(x))$ indicating that for all $x$, $\lambda(A(x)) \in \sigma^*(A(x))$.

\begin{figure}[hbt!]
    \centering
    \includegraphics[scale=.2]{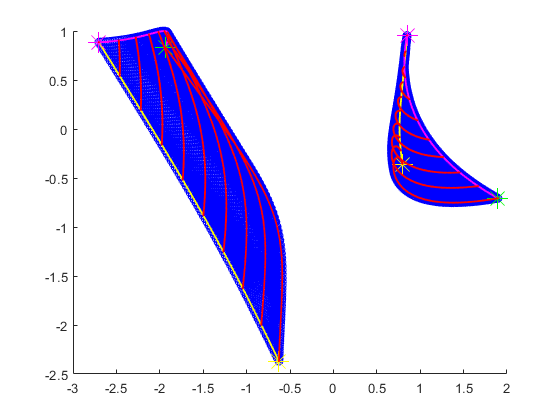}
    \includegraphics[scale=.2]{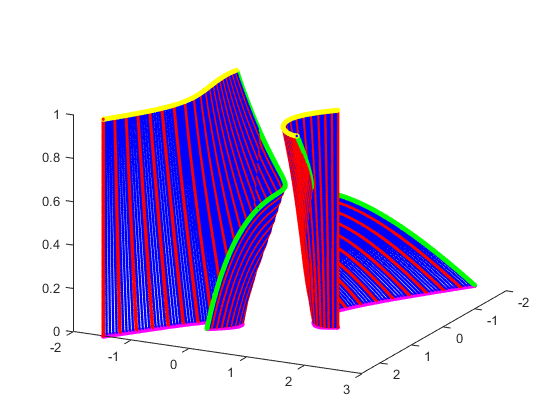}
    \caption{The eigen-surface of the convex-hull of three matrices projected into $\mathbb{C}$ (left) and represented in 3 dimensions (right).}
\end{figure}
We suppress our choice of spanning matrices in this notation as any choice will generate an eigensurface that is equivalent up to diffeomorphism\footnote{We prove this in the following section.}. It will be useful to think of this set as the collection of all eigenpaths realized by $S$. This object has particular significance in physics as it can be thought of as the potential energy surface realized by a particular atomic configuration space \cite{Physics1}. However, this application is typically restricted to the convex hull of Hermitian matrices, whereas here we will loosen this restriction and consider arbitrary matrices. Working in this more general setting will allow us to gain additional insight on this problem, as well as apply this research to other important questions such as applications in the study of eigenvalues realized by $DS_n$, the doubly-stochastic matrices. These matrices play an important role in probability and modeling, and they relate to this work as they can be represented as the convex hull of the permutation matrices \cite{MatAnal}.

\section{The Eigen-Surface}

Let $S\subseteq M_n$ be path-connected. We call the path-components of $ES(S)$ $k$-components, $\mathcal{K}$. We begin by making some basic observations on the structure of these sets.
\begin{lemma}
Let $\mathcal{K}$ be a path-component of\\ $ES(S)$. If $A\in S$ is such that there exists eigenvalues $\lambda_1(A),\dots, \lambda_k(A)$ where $\lambda_i(A)\in \mathcal{K}$ for all $i\in [k]$, then there are $k$ such eigenvalues (up to algebraic multiplicity) in $\mathcal{K}$ for any matrix $B\in S.$
\end{lemma}
\begin{proof}
For arbitrary $A,B\in S$, take a path in $S$ with initial matrix $A$ and terminal matrix $B$. Because there is an eigenpath that starts at each $\lambda_i(A)\in \mathcal{K}$ and ends at an eigenvalue of $B$, we see that $B$ then has at least as many eigenvalues as $A$ in $\mathcal{K}$. However, we can also reverse our path so that $A$ has at least as many eigenvalues as $B$ in $\mathcal{K}$, verifying our claim.
\end{proof}
\begin{corollary}
$ES(S)$ can be partitioned into distinct path-components $\mathcal{K}_1,\dots, \mathcal{K}_\ell$ where\\
$
\sum_{k=1}^\ell k_i = n.
$
\end{corollary}

Here, we briefly discuss our definition of the eigen-surface as being basis independent. We call a representation of $ES(S)$ the set constructed for a particular choice of $A_1,\dots,A_k$.

\begin{lemma}
        Any two representations of $ES(S)$ are equivalent up to a smooth diffeomorphism.
\end{lemma}

\begin{proof}
Any choice of $A_1,\dots,A_k$ as hypothesized will be a basis for $\operatorname{span}\{S\}$, that is, the set of all linear combinations of the elements in $S$. When we map one representation to the other, we only need to change the coordinates $\alpha_1,\dots,\alpha_k$, which corresponds to a change of basis, which is smooth.
\end{proof}

We can easily drop the condition that $A_1,\dots,A_k$ is a basis for $\operatorname{span}\{S\}$ (we will have to do this in section 4), because all of the following results will still hold, however the analysis of the eigen-surface in this case may be representation dependent. We can also replace the idea of using weights as coefficients and just append the matrix itself as the first element in the eigen-surface. This, however, is somewhat less practical as this will add many extra dimensions to our data. We can also make similar adjustments if we decide to parameterize $S$ by some other function (say if $S$ is a manifold) rather than just by linear coefficients.

\begin{lemma}
        If $S$ is compact then $ES(S)$ is compact. In particular, each $\mathcal{K}$ is compact.
\end{lemma}

\begin{proof}
We first show that they are bounded. Let $||-||$ be an arbitrary matrix norm. Then, as this norm is continuous, we see that $\sup_{A\in S}||A|| <\infty$. We then have that $\rho(A):= \max_{\lambda\in\sigma(A)} |\lambda| \leq ||A|| < \infty$. Because the coefficients $\alpha_1,\dots,\alpha_k$ also depend continuously on $S$, it follows that each $\mathcal{K}$ is bounded.\\
To show that each $\mathcal{K}$ is closed, we argue that a convergent sequence in $\mathcal{K}$ converges within the set. We write this sequence as $\{\lambda(A_n)\}$, where
$$
\lambda(A_n) = \{\alpha_1^n,\dots,\alpha_k^n, \lambda_n\}.
$$
We see then that the weights $\alpha_i^n$ correspond to the matrix
$$
A_n = \alpha_1^n A_1 + \dots + \alpha_k^n A_k
$$
and that each $\alpha_i^n\to \alpha_i$ as the component sequences must also converge, so that we get
$$
A_n \to A:= \alpha_1A_1 + \dots + \alpha_kA_k \in S
$$
by the compactness of $S$. What remains to show is that there is $\lambda \in \sigma(A)$ such that $\lambda_n\to \lambda$. This is immediate by the continuity of the spectrum.
\end{proof}

We now prove a result that will be central to the main theorem. Let $|-|$ be the euclidean metric on the appropriate space.
\begin{theorem} \label{KcompSep}
   Let $S\subseteq M_n$ be compact. There exists a neighborhood
   $$
   S_\delta:= \{v\in \mathbb{R}^k \times \mathbb{C}:\ \exists p \in S,\ |p-v|<\delta\}
   $$
   of $S$ such that $ES(S_\delta)$ has the same maximal $k$-components as $S$, that is, for $A,B\in S$, if $\lambda(A)\not\rightsquigarrow \lambda(B)$ in $S$ then $\lambda(A)\not\rightsquigarrow \lambda(B)$ in $S_\delta$.
\end{theorem}
\begin{proof}
    First we will take $\varepsilon>0$ such that for neighborhoods
    \begin{gather*}
    \mathcal{K}_{\varepsilon_n}:=\{z\in \mathbb{C}:\ \exists k\in \mathcal{K}_n,\ |z-k|<\varepsilon\},
    \end{gather*}
    we have for $i\neq j$,
    $$
    \mathcal{K}_{\varepsilon_i}\cap \mathcal{K}_{\varepsilon_j} = \emptyset.
    $$
    This is possible as each $\mathcal{K}_n$ is compact and disjoint. Viewing this in terms of our set $S$, we see that for $A\in S$, $\lambda_i(A)\in \mathcal{K}_i$, and $\lambda_j(A)\in \mathcal{K}_j$,  $|\lambda_i(A)- \lambda_j(A)|> \varepsilon$. Taking the projections (or eigenvalues associated with these points) into $\mathbb{C}$, $\lambda_i,\lambda_j$ respectively, we see then that $|\lambda_i-\lambda_j|>\varepsilon.$ By the continuity of the spectrum, this tells us that there exists $\delta_A>0$ so that if we perturb $A$ by a matrix $|B|<\delta_A$ (that is, take $A+B$), then these eigenvalues will still remain separated. In other words, $\mathcal{K}_i,\mathcal{K}_j$ will remain disjoint under small perturbations.\\
    We can cover $S$ by these balls of radius $\delta_A$, and form a neighborhood $U$ of our set $S$. We then can take a uniform neighborhood $S_\delta \subset U$ by the compactness of $S$, so the result is proven.
\end{proof}

\section{Transitivity}
We first introduce notation that will be useful to us in the following. For $A,B\in S$, $\lambda(A)\rightsquigarrow \lambda(B)$  if $\lambda(A), \lambda(B)\in \mathcal{K}$ for some $k$-component $\mathcal{K}\subseteq ES(S)$. We say that these eigenvalues are paired. In particular, this means that there is a path $A(x)\in S$ for all $x$, $A(0) = A,\ A(1) = B$ and an eigenpath $\lambda(A(x))$ where $\lambda(A(0)) = \lambda(A),\ \lambda(A(1)) = \lambda(B).$ \\
We say $S$ is transitive if whenever $\lambda_i(A)\rightsquigarrow \lambda_j(A)$ then $\lambda_i(A) = \lambda_j(A).$ We choose the word transitive because we can alternatively think of this property through the use of intermediary matrices.
\begin{lemma}
        $S$ is transitive if and only if for all $A,B\in S$, if $\lambda(A)\rightsquigarrow\lambda(B)$, then $\lambda(B)\rightsquigarrow \lambda(A)$ and this is the only distinct eigenvalue of $A$ which it is paired with.
\end{lemma}
\begin{proof}
By our condition, there is a path from $A$ to $B$ that pairs $\lambda(A),\lambda(B)$. We can join this with any path from $B$ to $A$ to create a closed loop from $A$ to $A$, so that by transitivity, we see that the same pairing must then be preserved.
\end{proof}
Transitivity then is the property that pairings are not path dependent. That is, when a pairing occurs via one path in $S$ between two matrices, then the same pairing must occur for every path in $S$ between these two matrices. 

\begin{lemma} \label{ForcedPair}
   For a matrix $A \in M_n$, there exist a neighborhood $N \subseteq M_n$ of $A$ and neighborhoods $D_i \subseteq \CC$ of the eigenvalues $\lambda_i\in \sigma(A)$ such that for all matrix paths $P(x)$ contained fully within $N$, any eigenpath of $P(x)$ is contained entirely within one $D_i$.
\end{lemma}

\begin{proof}
   If $A$ has only one distinct eigenvalue, then take any $\varepsilon > 0$ and choose $\delta > 0$ such that 
   $$
   \norm{A - B} < \delta \Longrightarrow 
   |\sigma(A) - \sigma(B)| < \varepsilon.\footnote{We use this notation to represent {\tiny $\displaystyle \max_{\lambda\in \sigma(A)}\min_{\mu\in\sigma(B)}|\lambda - \mu |$}}
   $$
   If $A$ has more than one distinct eigenvalue, define
   \[R:=\min_{\lambda_i(A)\neq \lambda_j(A)}|\lambda_i(A)-\lambda_j(A)|/2 \]
   and choose $\delta > 0$ such that 
   $$
   \norm{A-B} < \delta \implies |\sigma(A) - \sigma(B)| < R.
   $$
   For $i = 1, \ldots, n,$ let $D_i = \textbf{B}_R(\lambda_i)$\footnote{Here we use $\textbf{B}_\delta(A)$ to refer to the set $\{B\in M_n\ |\ |B-A|<\delta\}$, where $|-|$ denotes the Euclidean norm on $M_n$. We use a bold-faced $\textbf{B}$ here to distinguish the open ball from a matrix $B$.}. Choose any matrix path $P(x)$ contained within $\textbf{B}_\delta(A)$ and corresponding eigenpath $\lambda(x)$, and note that $\lambda(0)$ is in $D_k$ for some $k$. In fact, $\lambda(x)$ is contained in $D_k$ for all $x \in [0,1]$ since all $D_i$ corresponding to distinct eigenvalues are disjoint by choice of $R$, so that if $\lambda(x)$ leaves $D_k$, because it is continuous, there must be some point where it is at least distance $R$ from all eigenvalues of $A$, contradicting the choice of $\delta$.
\end{proof}
We now want to highlight a particular type of set that will be important in the study of transitivity. If $S$ has the property that for any $A,B\in S$, if $\lambda\in \sigma(A)$ has algebraic multiplicity $k$, then there exists $\mu\in \sigma(B)$ where $\mu$ has algebraic multiplicity $k$, we say that $S$ has uniform multiplicity. For example, if all matrices in $S$ have simple (algebraic multiplicity 1) eigenvalues, then $S$ will have uniform multiplicity.
\begin{corollary} \label{Notches}
Let $S\subseteq M_n$ have uniform multiplicity. For every $A\in S$, then there exists $\delta >0$ so that $\textbf{B}_\delta(A)\cap S$ is transitive.
\end{corollary}
\begin{proof}
We first note that by the continuity of the spectrum and our assumption of uniform multiplicity, we can choose $\delta>0$ so that each $D_i$ corresponding to $\textbf{B}_\delta(A) \cap S$ has eigenvalues of strictly 1 algebraic multiplicity, and thus each matrix has a unique eigenvalue in this neighborhood. We further construct these neighborhoods to be disjoint as in \cref{ForcedPair}, so we see that $\lambda(B)\rightsquigarrow \lambda(B)$ and this is the only such pairing, thus it is transitive.
\end{proof}

\begin{figure}[hbt!] \label{BlobNotch}
    \centering
    \includegraphics[scale=.4]{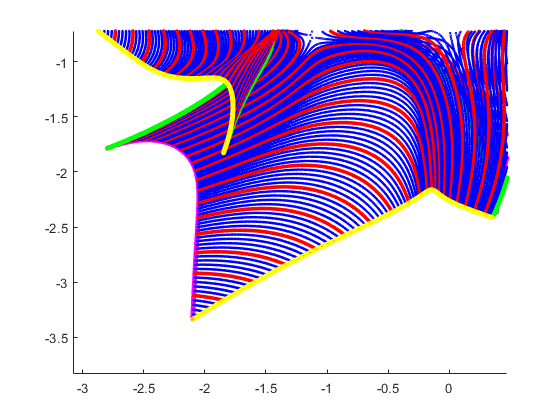}
    \caption{A 3-dimensional representation of an eigen-surface that is not transitive, but regions of transitivity can be observed close to its boundary.}
\end{figure}

We will use these tools to study what happens when we bend our matrix paths, with particular interest in what causes a change in the pairing of eigenvalues.\\
For closed paths $A_1(x):[0,1]\to M_n,A_2(x):[0,1]\to M_n$, we call a deformation of $A_1(x)$ and $A_2(x)$ in $S$ the continuous function 
$$
A(x,y):[0,1]\times [0,1]\to M_n
$$
where $A(x,0) = A_1(x)$, $A(x,1)=A_2(x)$, and $A(x,y)$ for fixed $y$ is a closed loop in $S$. Accompanying this, we develop new notation to describe the eigenpaths of this multivariate function. We will take $\lambda_i^x(x,y)$ to be the eigenpath induced by the path $A(x,y)$ for fixed $x$, and $\lambda_i^y(x,y)$ to be the eigenpath induced by the path $A(x,y)$ for fixed $y$. To describe the pairing of eigenvalues by a particular path, or eigenpath $p(x)$, we will use the notation $\lambda \rightsquigarrow_{p(x)} \mu.$
\begin{lemma}\label{split}
        Let $A(x):[0,1]\to M_n$, and $\lambda\rightsquigarrow_{A(x)} \mu_1,\mu_2\in \sigma(A(1)),$ where $\mu_1\neq \mu_2$. Then $A(x)$ cannot have uniform multiplicity.
\end{lemma}
\begin{proof}
Let $\lambda_1(x),\dots,\lambda_n(x)$ be the eigenpaths induced by $A(x)$. Suppose that $\lambda_i(0) = \lambda_j(0)$, but $\lambda_i(1)\neq \lambda_j(1)$, as hypothesized. As these paths are continuous, there must be some $0\leq L\leq 1$ that is the maximum value of $x$ which $\lambda_i(x)=\lambda_j(x).$ What we want to now show is that when this higher multiplicity eigenvalue is lost, it cannot be replaced by combining lower multiplicity eigenvalues.
\\
To show this, consider $\lambda_h(x),\lambda_k(x)$ where $\lambda_h(L) \neq \lambda_k(L)$. Again, by the continuity of these paths there is some interval about $L$ where $\lambda_k(x)\neq \lambda_h(x)$, as was desired.
\end{proof}
\begin{theorem} \label{HomotopyPairs}
   Let $A_1(x),A_2(x)$ be closed paths in $S$ and let $A(x,y)$ be a deformation of $A_1$ and $A_2$ in $S$. If for any $i,j$ we have 
   \begin{gather*}
   \lambda_i(0,0)\rightsquigarrow_{\lambda_i^y(x,0)} \lambda_j(1,0),\\
   \lambda_i(0,0) \rightsquigarrow_{\lambda_i^x(0,y)} \lambda_i(0,1),\\
   \lambda_i(1,0) \rightsquigarrow_{\lambda_i^x(0,y)} \lambda_i(1,1),
   \end{gather*}
   but
   $$
   \lambda_i(0,1)\not\rightsquigarrow_{\lambda_i^x(x,1)}\lambda_j(1,1)
   $$
    then the image of $A(x,y)$ cannot have uniform multiplicity. 
\end{theorem}
The last condition somewhat defies our intuition. In some sense, it feels like all of these eigenvalues which we can continuously traverse from one to the other should be the same, and that pairing should be preserved under these deformations. The statement of this theorem is that when this intuition is violated, something must go wrong in our set of matrices $S$, and we formalize this notion of ``going wrong" through the use of uniform multiplicity or the lack thereof.
\begin{proof}
Let $\lambda_i(y)$ be the eigenpath induced by $A(0,y)$, with initial eigenvalue
$$
\lambda_i(0)=\lambda_i\in \sigma(A(0)).
$$
Define the set
    $$
    \mathcal{Y}:=\{y \in [0,1]\ |\ \lambda_i(y) \rightsquigarrow_{\lambda_i^y(x,y)} \lambda_j(y)\},
    $$
and $L:=\sup \mathcal{Y}$ where $L\geq 0$ by hypothesis. Assume that $S$ has uniform multiplicity for the sake of contradiction. By corollary \ref{Notches}, we see that $L>0$, as $S$ is locally transitive. Similarly, we see that $L<1$, as by hypothesis $\lambda_i(1) \not \rightsquigarrow_{\lambda_i^y(x,y)} \lambda_j(1)$.
\\
Let $C$ be the image of $A(x,L)$. By theorem \ref{KcompSep}, we see that this set has locally maximal pairings, that is, there exists some $\delta>0$ so that the neighborhood $C_\delta$\footnote{As before, we take this to be the uniform $\delta$ neighborhood of the set.} has the same $k$-component structure. What we now want to show is that approaching $C$ from above and below will give us different pairings, violating the uniform multiplicity condition.
\\
We first look at pairings from below ($y\leq L$), by looking at two cases, if $L\in \mathcal{Y}$ and $L\not\in \mathcal{Y}$. In the first of these, by definition $\lambda_i(L)\rightsquigarrow_{\lambda_i^y(x,L)}\lambda_j(L)$. We want to show that the second case is impossible. By construction, we can take a sequence $\{\phi_n\}$, $\phi_n\to L$, where $\lambda_i(\phi_n)\rightsquigarrow_{\lambda_i^y(x,L)}\lambda_j(\phi_n)$. As $A(x,y)$ is continuous on a compact domain, it is uniformly continuous, so there exists $N>0$ where for $n\geq N$, $|A(x,\phi_n)-A(x,L)|<\delta$, that is, $A(x,\phi_n)\in C_\delta.$ Thus we see that 
\begin{gather*}
\lambda_i(L)\rightsquigarrow_{\lambda_i^x(0,y)} \lambda_i(\phi_N)\\
\rightsquigarrow_{\lambda_i^y(x,\phi_N)} \lambda_j(\phi_N)\rightsquigarrow_{\lambda_i^x(1,y)} \lambda_j(L)
\end{gather*}
in other words, there is a path in $C_\delta$ so that $\lambda_i(L)\rightsquigarrow\lambda_j(L)$, contradicting that $L\not\in \mathcal{Y}.$
\\
We want to argue similarly for pairings from above ($y\geq L$). If we can show that there is a sequence $\{\phi_n\}$, $\phi_n\to L$, $\lambda_i(\phi_n)\rightsquigarrow \lambda_h(\phi_n)$ where $\lambda_h(\phi_n)\neq \lambda_j(\phi_n)$, then as before, we will be able to show $\lambda_i(L)\rightsquigarrow \lambda_k(L)$. This, however, is clear, as for all $y\geq L$, $\lambda_i(y)\not\rightsquigarrow \lambda_j(y)$, and there are only finitely many other options for pairing, so at least one of them must occur infinitely often locally about $L$.
\\
In summary, we have shown that there exists $j\neq h$ so that $\lambda_i(L)\rightsquigarrow \lambda_j(L), \lambda_h(L)$ in $C$. $C$, however, is just the path $A(x,L)$, so it must be that $\lambda_i(L)\rightsquigarrow_{\lambda_i^y(x,L)} \lambda_j(L),\lambda_h(L)$ so by lemma \ref{split}, $S$ cannot have uniform multiplicity, giving us a contradiction.
\end{proof}
\begin{corollary}
   Any continuous deformation of a path over a space of uniform multiplicity will preserve pairings.
\end{corollary}

\begin{corollary} \label{NonTrans}
   Let $S$ be a simply connected subset of $M_n$, and let $A(x)$ be a closed path in $S$.
   If the image of $A(x)$ is not transitive, then not all matrices in $S$ can have the same multiplicity lists.
\end{corollary}
\begin{proof}
This follows by taking the contraction from $A(x)$ to the initial point $A(0)$. As $A(0)$ is trivially transitive, a different pairing must occur in $A(x)$, and so by \cref{HomotopyPairs} $A(x,y)$ cannot have uniform multiplicity, and thus $S$ cannot have uniform multiplicity.
\end{proof}
\begin{corollary} \label{SimplyConnectedTrans}
A simply connected set $S$ is transitive if and only if it has uniform multiplicity.
\end{corollary}

\begin{proof}
Clearly, transitivity implies uniform multiplicity. Further, by \cref{NonTrans}, uniform multiplicity implies that all paths in $S$ must have the same pairings, and because the trivial constant paths are transitive, $S$ must be transitive.
\end{proof}
\begin{corollary}
If $S$ is simply connected, and $\mathcal{K}$ is one of its $k$-components, $k>1$, then there must be an eigenvalue with algebraic multiplicity $>1$ in $\mathcal{K}.$
\end{corollary}
\begin{figure}[hbt!] \label{SaddlePoint}
    \centering
    \includegraphics[scale=.4]{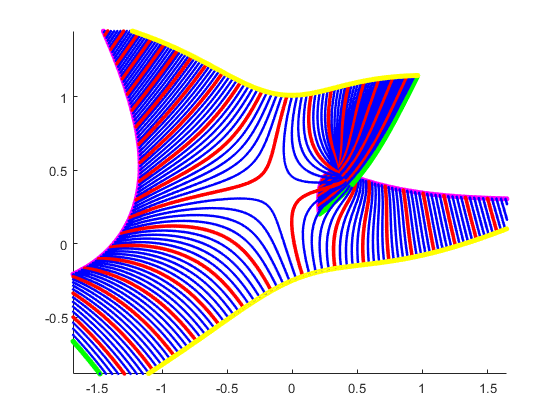}
    \caption{A 3-dimensional representation of an eigen-surface that is not transitive. The ``saddle point" in the center of the figure is indicative of a higher multiplicity eigenvalue.}
\end{figure}

\section{Classifying Eigenvalues}
We now restrict ourselves to the convex-hull of matrices, that is, for matrices $A_1,\dots,A_\ell\in M_n$,
\begin{gather*}
    Co(A_1,\dots,A_k):= \{A\in M_n\ |\ A = \sum_{i=1}^k \alpha_iA_i,\\ \sum_{i=1}^k \alpha_i = 1, \forall i\ \alpha_i \in [0,1] \},
\end{gather*}
which is simply-connected and compact. We will no longer require $A_1,\dots,A_k$ to necessarily be linearly independent, opting instead for the natural coefficients provided by taking the convex hull.
\\
To each matrix $A\in M_n$, we associate to it a multiplicity list
$$
(a_n,\dots,a_1)
$$
where $a_i$ is the number of distinct eigenvalues of $A$ with algebraic multiplicity $i$. If $A,B\in M_n$, then we say that $A \overline{\gtrdot} B$ if the multiplicity list of $A$ is greater than or equal to the multiplicity list of $B$ in the lexicographic ordering, and we omit the bar for strict inequality. In other words, we are taking these $a_i$ and comparing them in descending order with respect to $i$, for example,
$$
\begin{bmatrix}
2 & 0 & 0\\
0 & 2 & 0\\
0 & 0 & 1
\end{bmatrix}
\gtrdot 
\begin{bmatrix}
3 & 0 & 0\\
0 & 2 & 0\\
0 & 0 & 1
\end{bmatrix}.
$$
We start by using methods with the polynomial discriminant to study the multiplicity lists of matrices related by simple paths. We can describe any 
$$
A\in Co(A_1,\dots,A_k)
$$
by its coefficients 
$$
A_{cof}:=(\alpha_1,\dots,\alpha_{k-1})
$$ 
where we only need $k-1$ coefficients as one is determined by the others. Additionally, as there are only finitely many multiplicity lists of length $n$, we can define $\mathcal{C}$, the core of $Co(A_1,\dots,A_k)$ to be the collection of matrices that have the smallest multiplicity lists in the lexicographic ordering. We define $\mathcal{U}:= Co(A_1,\dots,A_k)\setminus \mathcal{C}.$ We call a convex hull regular if the minimal multiplicity list is $(0,0,\dots,0,n)$, so that any matrix $A\in\mathcal{C}$ is simple.
\subsection{Regular Convex Hulls}
We begin with a treatment of a somewhat restricted case as there is a wider variety of tools at our disposal. In particular, we will be able to use the discriminant to identify whether matrices are in $\mathcal{U}$ or $\mathcal{C}$, giving us additional information on the structure of the eigen-surface.
\begin{lemma} \label{ZeroSet}
    If $S$ is regular then there is a holomorphic function $f$ of $k-1$ variables such that $f(A_{cof})=0$ if and only if $A\in\mathcal{U}$.
\end{lemma}
\begin{proof}
Let $p :M_n \to \mathbb{P}_n$, $p(A)=\det(A-\lambda I)$ denote the map from a matrix to its characteristic polynomial. This is holomorphic as it is polynomial in the entries of $A$. Further, we can then take $f:\mathbb{R}^{k-1}\to M_n$ which is the map $f(A_{cof}) = A$, which is also holomorphic, as $A$ is defined as the linear combination
\begin{gather*}
    A=\alpha_1A_1+\dots + \alpha_{k-1}A_{k-1}\\
    + (1-\alpha_1- \dots - \alpha_{k-1})A_k.
\end{gather*}
Finally, we define $d:\mathbb{P}_n \to \mathbb{C}$ to be discriminant, where we consider the characteristic polynomial as a polynomial in $\lambda$. This then gives us a polynomial in the coefficients of the input polynomial, thus
$$
d\circ p \circ f
$$
is a polynomial in our coefficients $\alpha_1,\dots,\alpha_{k-1}$. It satisfies the desired properties of our mapping, as it is zero if and only if $A$ has an eigenvalue of multiplicity greater than 1.
\end{proof}

\begin{corollary}
   If $S=Co(A_1,\dots,A_k)$ is regular, then $\mathcal{C}$ is open in the Zariski topology.
\end{corollary}
This immediately implies the following lemma, but we also offer an elementary proof.
\begin{lemma} \label{Core}
$\mathcal{C}$ is dense in $Co(A_1,\dots,A_k)$. Additionally, it is open in the subspace topology.
\end{lemma}
\begin{proof}
The fact that it is open is clear by the continuity of the spectrum. That is, the algebraic multiplicity of the eigenvalues of a matrix $C\in\mathcal{C}$ cannot increase under a small enough perturbation.
\\
To show that it is dense, it is sufficient to show that for arbitrary $A\in S$, for all $\varepsilon>0$, there exists $C\in\mathcal{C} \cap \textbf{B}_\varepsilon(A)$. We argue by contradiction. Suppose not, then there exists $\textbf{B}_\varepsilon(A)$ such that for all $B\in \textbf{B}_\varepsilon(A),$ $B\gtrdot C$ for any $C\in \mathcal{C}$. If we take the line segment $A(x)$ from $A$ to $C$, we see that $f(A(x)_{cof})\equiv 0$ in this neighborhood, but it is analytic so it must be that $f(A(x)_{cof}) \equiv 0$ for all $x\in [0,1]$. But then $A(x)\in \mathcal{U}$ for all $x$, giving us a contradiction as $C\in \mathcal{C}$.
\end{proof}
\begin{lemma}
   $Co(A_1,\dots,A_k)$ is almost {\small everywhere} locally transitive. Further, we can partition all but an arbitrarily small subset of $Co(A_1,\dots,A_k)$ into finitely many transitive, simply connected sets. 
\end{lemma}
\begin{proof}
The first statement follows from the zero set of a multi-variable analytic function being negligible, so that $Co(A_1,\dots,A_k)=\mathcal{C}$ almost everywhere. For the second statement, by \cref{Core}, for each matrix $C\in \mathcal{C}$, we take $\varepsilon>0$ such that $\textbf{B}_\varepsilon (\mathcal{C}) \cap \mathcal{U} = \emptyset$. Then, for arbitrary $\delta>0$, we can take the $\delta$-neighborhood $\mathcal{U}_\delta$ of $\mathcal{U}$. These open sets form a cover of $Co(A_1,\dots, A_k)$, and because it is compact, we can take a finite subcover. All of the sets except for $\mathcal{U}_\delta$ are transitive by \cref{SimplyConnectedTrans}, and it is clear that we can form a partition of the convex hull by taking unions and intersections of these sets as necessary.
\end{proof}
We now will recall some critical past results in perturbation theory. From \cite{Kato}, and \cite{Knopp}, we see that in the case where
$
S=Co(A_1,\dots,A_k),
$
and $S$ has uniform multiplicity, then we can not only parameterize our set by the holomorphic function
$$
A(\alpha_1,\dots,\alpha_k) = \alpha_1A_1 + \dots + \alpha_k A_k
$$
but we also have that there are no so-called \textit{exceptional points} where two initially distinct eigenvalues combine to form an eigenvalue of higher algebraic multiplicity. Thus we can parameterize our eigenvalues as holomorphic functions 
$$
\lambda(\alpha_1,\dots,\alpha_k),
$$
so that $S$ and each $\mathcal{K}$ are conformal in this case.
\begin{corollary}
$ES(\mathcal{C})$ is a smooth manifold with boundary.
\end{corollary}
We have now identified a very regular and well-behaved portion of the eigen-surface, but we still need to address $\mathcal{U}$. We classify this collection by the local transitivity properties of each matrix.\\
For a matrix $A\in \mathcal{U}$, $A$ is locally non-transitive about $\lambda(A)$ if, for any $\varepsilon>0$, $B\in \textbf{B}_\varepsilon(A)\cap \mathcal{C}$, if
$$
\lambda_i(B),\lambda_j(B)\rightsquigarrow \lambda(A)
$$
then $\lambda_i(B)\rightsquigarrow \lambda_j(B)$ in $\textbf{B}_\varepsilon(A)\cap \mathcal{C}$. What this definition signifies is that there is locally a path in $\mathcal{C}$ which connects all of the eigenvalues which are path-connected through $A$. We say these eigenvalues collide at $A$.
\\
What is key to note in this definition is that $\mathcal{C}$ itself will not always be transitive. As was shown in corollary \ref{SimplyConnectedTrans}, a sufficient condition for that to occur is for $\mathcal{C}$ to be simply connected.

\begin{figure}[hbt!] \label{SaddlePoint}
    \centering
    \includegraphics[scale=.4]{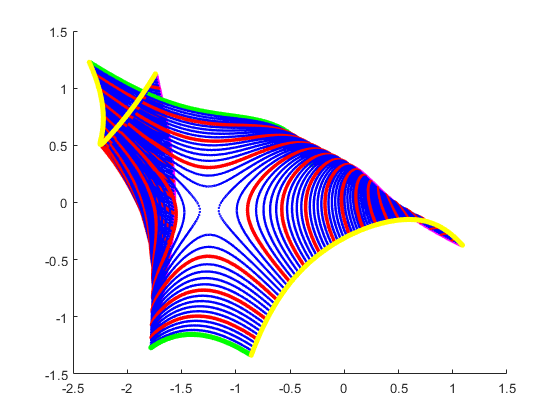}
    \caption{An example of a $Co(A_1,A_2,A_3)$ where $\mathcal{U}$ is a single matrix that is locally non-transitive.}
\end{figure}

\begin{lemma}
If $A$ is locally non-transitive about $\lambda(A)$ then there exists $\varepsilon>0$ such that $$ \textbf{B}_\varepsilon(\lambda(A))\cap ES(Co(A_1,\dots,A_k))$$ is topologically a manifold. In particular, it is homeomorphic to 
$$
\textbf{B}_\varepsilon(A)\cap Co(A_1,\dots,A_k).
$$
\end{lemma}
\begin{proof}
It suffices to show that we can parameterize these eigenvalues as a single continuous function. By hypothesis, we can take a closed path $B(x)\in \textbf{B}_\varepsilon(A)\cap \mathcal{C}$ for all $x$, where for $B:=B(0)$ there is an eigenpath $\lambda(B(x))$ that achieves all eigenvalues which collide at $\lambda(A).$ By \\\cref{ZeroSet}, we see that for sufficiently small $\varepsilon>0$, $\mathcal{C}\cap\textbf{B}_\varepsilon (A)$ forms the same homotopy equivalence class, thus by \cref{HomotopyPairs} we can continuously parameterize all eigenvalues realized by this set in this way. We can extend this function continuously to $\lambda(A)$ and the other matrices in $\mathcal{U}\cap \textbf{B}_\varepsilon(A)$.
\end{proof}

We observe that in this case, our eigenvalues will no longer be holomorphic, as described in \cite{Kato}.
\\
This gives us a sufficient condition for a matrix $A\in \mathcal{U}$ to not significantly disrupt the structure of the eigen-surface, preserving the topological structure of $ES(\mathcal{C}).$ To better understand the remaining eigenvalues, it helps to look more closely at the zero set of this polynomial.

\begin{lemma}
        Generically, $\mathcal{U}$ will be no more than codimension 2 in $ES(Co(A_1,\dots,A_k))$.
\end{lemma}

\begin{proof}
Recall that $\mathcal{U}$ is essentially the zero set of a polynomial $d$ in the variables $\alpha=(\alpha_1,\dots,\alpha_{k-1})$. These inputs are strictly real, while $d$ has complex coefficients. Thus we can separate $d$ into real and complex parts
\begin{gather*}
d(\alpha) = p(\alpha) + i q(\alpha),\\
p(\alpha):= \operatorname{Re}(d(\alpha)),\quad q(\alpha):=\operatorname{Im}(d(\alpha)),
\end{gather*}
where each function is polynomial. $d(\alpha)=0$ if and only if $p(\alpha) = q(\alpha) = 0.$ If $p$ and $q$ are non-constant polynomials, then each has zero set of (real) codimension 1, and the intersection of their zero sets has codimension 1 if and only if $p$ and $q$ share a factor in $\mathbb{C}[\alpha]$.
\end{proof}

\begin{theorem}
Let $Co(A_1,\dots,A_k)$ be regular. If $A\in \mathcal{U}$ is not locally non-transitive about $\lambda(A)$, where $\lambda(A)$ has multiplicity $>1$, and $\mathcal{U}$ has codimension 2 or fewer, then there is no $\varepsilon>0$ such that 
$$ 
\textbf{B}_\varepsilon(\lambda(A))\cap ES(Co(A_1,\dots,A_k))
$$ 
is topologically a manifold. 
\end{theorem}

\begin{proof}
It suffices to show that there is no local chart to Euclidean space at $A$. By assumption, $\mathcal{U}$ is a manifold with boundary of at most $k-3$ real dimensions, while $Co(A_1,\dots,A_k)$ is a $k-1$ dimensional manifold with boundary. We see however that
$$ 
\textbf{B}_\varepsilon(\lambda(A))\cap ES(Co(A_1,\dots,A_k)) \setminus ES(\mathcal{U})
$$ 
is locally disconnected by hypothesis, which is not possible if $ES(Co(A_1,\dots,A_k))$ is a $k-1$ dimensional manifold with boundary.
\end{proof}
\begin{figure}[hbt!] \label{Separable}
    \centering
    \includegraphics[scale=.4]{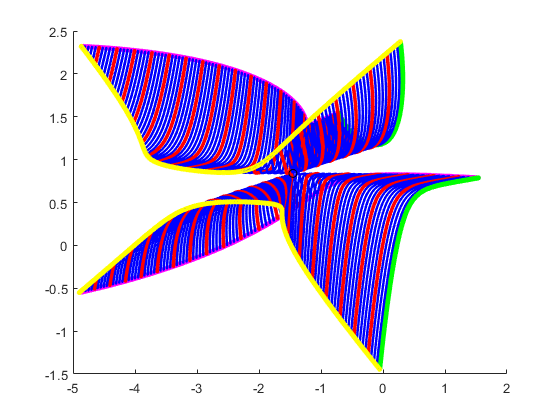}
    \caption{A 3-dimensional representation of $ES(Co(A_1,A_2,A_3))$ which is not a manifold with boundary, as there is a matrix $A\in \mathcal{U}$ which is diagonalizable.}
\end{figure}

We now try to identify local non-transitivity with structural properties of the matrices themselves. Another way of phrasing this property is that, locally, several eigenvalues cannot be parameterized individually, rather each one forms a branch about some higher multiplicity eigenvalue. On the other hand, by \cite{Kato}, we see that if $A\in \mathcal{U}$ has an eigenvalue of geometric multiplicity $m$, then locally at least $m$ of these eigenvalues can be parameterized as well-defined continuous functions. Thus $A$ cannot be locally non-transitive if its higher multiplicity eigenvalue also has higher geometric multiplicity. As discussed in \cite{Arnold}, this can alternatively be viewed as an increase in the dimension of the commutator space of $A$ relative to matrices in $\mathcal{C}$. We summarize this in the following.

\begin{corollary}
If $A\in \mathcal{U}$ has higher dimensional commutator space than all matrices in $\mathcal{C}$, then it is not locally non-transitive.
\end{corollary}

\begin{corollary}
   If $ES(Co(A_1,\dots,A_k))$ is topologically a manifold with boundary then it has uniform commutator space dimension.
\end{corollary}

\subsection{General Convex Hulls}

Everything will follow as before if we can similarly establish that the matrices in $\mathcal{U}$ correspond to the zero set of some polynomial. If $Co(A_1,\dots,A_k)$ is not regular, the discriminant will not be useful for us as it will be identically 0, thus we develop alternate methodology. 

\begin{lemma}
        Let $(a_n,\dots, a_1)$ be minimal where $i$ is the largest index such that $a_i\neq 0$. Define $p(\lambda, \alpha)$ (which is a polynomial in $\lambda$) to be the characteristic polynomial, and $(a_{n_\alpha},\dots,a_{1_\alpha})$ to be the multiplicity list of each matrix as they vary in $\alpha$. Then $a_{j_\alpha}>0$ for $j>i$ if and only if the polynomials $p(\lambda,\alpha)$ and 
        $$
        p^{(i+1)}(\lambda,\alpha):= \frac{\partial^{i+1}}{\partial \lambda^{i+1}} p(\lambda,\alpha)
        $$
        have a common zero.
\end{lemma}
Thus we are looking at the intersection of the zero set of two polynomials, which corresponds to a closed set in the Zariski topology. What we now want to do is induct our way downward on the multiplicity list of these polynomials. If we can see when each individual coordinate of the list increases, then we will have established a general criteria. To do this, it suffices to show that we can factor our characteristic polynomial so that we isolate all of the lower multiplicity eigenvalues. We see that in the cases where each eigenvalue is an entire holomorphic $r_i(\alpha)$, this is immediate as we have
$$
p(\lambda, \alpha) = \prod_{i=1}^n (\lambda - r_i(\alpha))
$$
so we can separate the roots as desired. If such a parameterization is not possible for some collection of roots, then it must be that there are subcollections of these roots which correspond to the branches of a holomorphic function, with a common branch point $c\in \mathbb{C}.$ Let $r_{i_1},\dots, r_{i_a}$ be such a subcollection. It then follows that
$$
q:=\prod_{j=1}^a (\lambda - r_{i_j})
$$
is holomorphic, and if $r_{i_j}$ is a root of multiplicity $m$, then all distinct roots in this subcollection must also be of multiplicity $m$. It follows then that we can achieve the desired factorization
$$
p(\lambda, \alpha) = q(\lambda, \alpha) r(\lambda,\alpha)
$$
where $q$ is the product of all highest multiplicity terms $\lambda-r_i$.
\begin{lemma}
        Let $(a_n,\dots, a_1)$ be minimal where $i$ is the largest index such that $a_i\neq 0$. If $i>0$, then we can factor 
        $$
        p(\lambda,\alpha)= q(\lambda,\alpha) r(\lambda, \alpha)
        $$
        where $q,r$ are holomorphic and $r$ has no permanent zeroes of multiplicity $i$ or greater.
\end{lemma}

\begin{corollary} \label{ZeroSet}
    The set $\mathcal{U}$ is closed in the Zariski topology.
\end{corollary}

\section{Computational Methods}

Computationally, it is impossible to consider all paths in $Co(A_1,\dots, A_k)$, thus we provide justification to study a simpler, more amenable collection of paths. We define $ES^\ell(S)$ to be the eigen-surface of a set $S$ where we only consider eigenvalue pairings by polygonal paths in $S$, that is, paths that are the union of line segments.

\begin{lemma} \label{GraphDist}
All pairings can be achieved by polygonal paths.
\end{lemma}

\begin{proof}
It is sufficient to show that for arbitrary $A\in Co(A_1,\dots,A_k)$, if $\lambda_i(A)\rightsquigarrow \lambda_k(A)$ in the image of an arbitrary path, then it also does by a polygonal path. By \cref{HomotopyPairs}, pairings are entirely determined by path homotopy class about $\mathcal{U}$ and intersections with $\mathcal{U}$. It is clear that we can construct a polygonal path that intersects the same components of $\mathcal{U}$, and we can also construct it to be in the same homotopy class about each component $\mathcal{U}$ as $\mathcal{C}$ is open as a subspace of $Co(A_1,\dots,A_k)$, so it follows that path connected $\Longleftrightarrow$ polygonal path connected.
\end{proof}
Define $\operatorname{ord}(\lambda_i(A))$ to be the number of $\lambda_j(A)\neq \lambda_i(A)$ such that $\lambda_i(A)\rightsquigarrow \lambda_j(A)$.
\begin{corollary}
    $\operatorname{ord}(\lambda(A))\leq |\pi_1(ES(S)\setminus ES(\mathcal{U}))|$. That is, the fundamental group of the eigen-surface determines the possible paths between eigenvalues.
\end{corollary}

This argument tells us that it is sufficient to consider only representative paths of each homotopy class, reducing the problem to a finite computation. For a $k$-component $\mathcal{K}\subseteq ES(S)$, we can define a graph $G(\mathcal{K})$ where $V(G(\mathcal{K})) := \{ \lambda(A)\in \mathcal{K}\}$ and $E(G(\mathcal{K})) := \{(\lambda(A),\lambda(B)) \in \mathcal{K}\times \mathcal{K}\ |\ \lambda(A)\rightsquigarrow \lambda(B)\}$. Without restricting to a particular type of path, this graph has little meaning, as each $G(\mathcal{K})$ will be complete, however, if we consider $G(\mathcal{K}^\ell)$ where only pairings induced by a line-segment correspond to an edge, we can deduce some less trivial properties.

\begin{lemma}
   $G(\mathcal{K}^\ell)$ has finite diameter for $\mathcal{K}^\ell \subseteq ES^\ell (Co(A_1,\dots, A_k))$
\end{lemma}

\begin{proof}
We argue by contradiction. We define
$$
d(\lambda(A),\lambda(B)):V(G(\mathcal{K}^\ell))\times V(G(\mathcal{K}^\ell)) \to \mathbb{Z}
$$
to be the maximum path length between two vertices in $G(\mathcal{K}^\ell)$ and we call $S:= Co(A_1,\dots,A_k)$. By \cref{GraphDist}, $d(\lambda(A),\lambda(B)) = c < \infty$ for any eigenvalues in $\mathcal{K}^\ell.$ If $G(\mathcal{K}^\ell)$ were to have infinite diameter, then there must be a sequence of matrices $(A_n,B_n)\in S\times S$ where
$$
\lim_{n\to\infty}d(\lambda(A_n),\lambda(B_n)) = \infty.
$$
Now, we can take a convergent subsequence of these tuples, where
\begin{gather*}
\lim_{n_i\to\infty}(A_{n_i},B_{n_i}) =: (A,B)\in S\times S,\\ \lim_{n_i\to\infty}(\lambda(A_{n_i}),\lambda(B_{n_i})) =: (\lambda(A),\lambda(B)).
\end{gather*}
 By \cref{ForcedPair}, for $n_i$ sufficiently large, the line segment from $A_{n_i}$ to $A$ and from $B_{n_i}$ to $B$ will induce pairing between $\lambda(A_{n_i})$ and $\lambda(A)$, and $\lambda(B_{n_i})$ and $\lambda(B)$. We have then that
 $$
 d(\lambda(A_{n_i}), \lambda(B_{n_i})) \leq d(\lambda(A), \lambda(B)) + 2
 $$
for all $n_i$ large enough. However, we also have that $d(\lambda(A_{n_i}), \lambda(B_{n_i}))$ is arbitrarily large for $n_i$ large enough, giving us a contradiction.
\end{proof}

We define the principle graph of $Co(A_1,\dots,A_k)$ to be the graph generated only by considering representative matrices in $\mathcal{U}$ and the matrices $A_1,\dots,A_k$. This object provides information on the overall structure of its eigen-surface, highlighting regions of high transitivity, and how these regions are joined together as a partition of \\$Co(A_1,\dots,A_k).$  

\begin{figure}[hbt!]
    \centering
    \includegraphics[scale=.2]{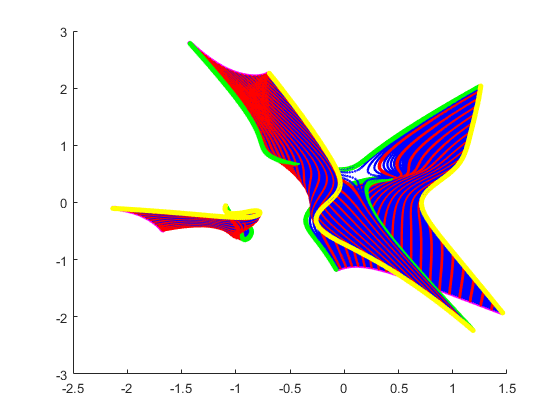}
    \includegraphics[scale=.2]{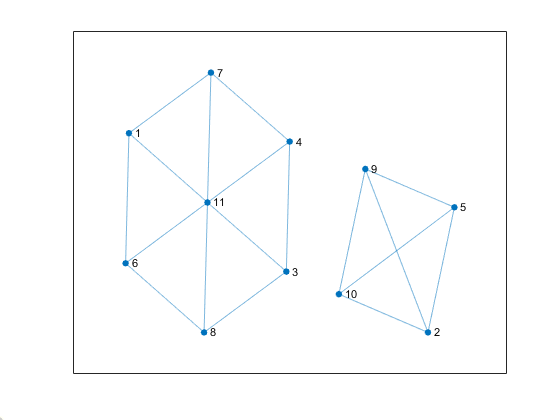}
    \includegraphics[scale=.2]{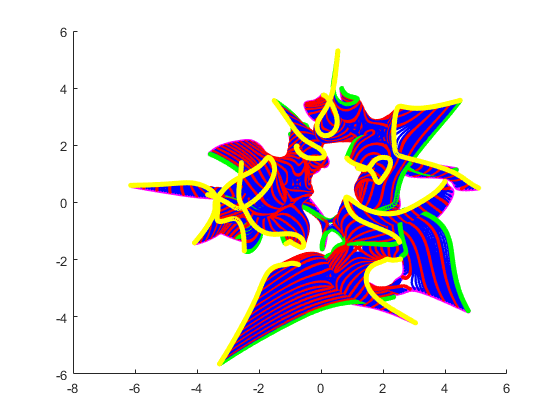}
    \includegraphics[scale=.2]{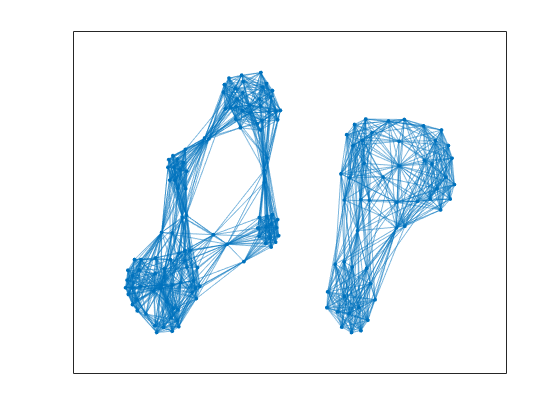}
    \caption{The principle graphs of two eigen-surfaces, each one corresponding to the convex hull of 3 matrices. The first two images show a 2-component and a 1-component easily identified as a wheel graph and a $K_4$ (a complete graph on 4 vertices), while the second two images are much more complex, and the graph only shows the general structural properties.}
\end{figure}

\section{Examples}

\subsection{Hermitian Matrices}

Hermitian matrices have real eigenvalues, and we can use this fact to show that all Hermitian matrix closed paths are weakly transitive, that is, for all $i\in [n]$, we can take paths $\lambda_i(A(x))$ where $\lambda_i(A(0)) = \lambda_i(A(1))$. This proof also holds for any other path connected set where the matrices have real eigenvalues.

\begin{lemma}
Let $S\subseteq M_n$ be a path connected set. If $S$ is such that all $A\in S$, $A$ is Hermitian $\Longrightarrow$ $S$ is weakly transitive.
\end{lemma}

\begin{proof}
Suppose that $\lambda_i(A)\rightsquigarrow\lambda_j(A)$ and $\lambda_k(A)$\\$\rightsquigarrow\lambda_i(A)$, $i\geq j, k$, with respect to a Hermitian matrix path. We see then that we can take continuous real functions\footnote{here we consider the eigenpaths as real functions} $\lambda_i(A(x)),\ \lambda_j(A(x)),$ where $\lambda_i(A)= \lambda_i(A(0)) \geq \lambda_j(A(0)) = \lambda_k(A),$ and $\lambda_j(A)$\\$ = \lambda_i(A(1)) \leq \lambda_j(A(1)) = \lambda_i(A)$. We see then by the intermediate value theorem that these paths must intersect, and so we can construct the desired path. We can argue similarly when considering the reverse inequalities.
\end{proof}

We further note that as the convex hull of Hermitian matrices is everywhere diagonalizable, no matrices in $\mathcal{U}$ will be locally non-transitive, so generically its eigen-surface fails to be a topological manifold if it does not have uniform multiplicity. These results also hold when considering the singular values of a convex hull, as these are the eigenvalues of Hermitian matrices.

\subsection{Nonnegative Irreducible Matrices}
Suppose $A_1, \ldots, A_k \in M_n$ are primitive, irreducible matrices. Then any convex combination thereof is also primitive irreducible. Thus, by the Perron-Frobenius theory, every element of $Co(A_1, \dots, A_k)$ has a Perron root, a real positive eigenvalue of highest magnitude amongst all eigenvalues, that is of multiplicity $1$ \cite{Perron}. As the spectrum is continuous, we see then that the Perron roots of matrices in $Co(A_1, \ldots, A_k)$ must all be contained in a single $1-$component.

\begin{figure}[hbt!]
    \centering
    \includegraphics[scale=.4]{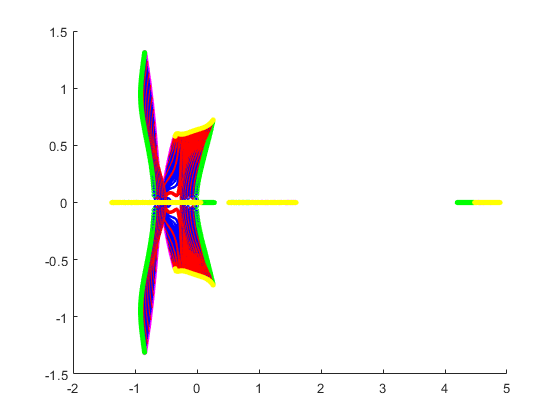}
    \caption{The eigen-surface of the convex hull of three nonnegative irreducible matrices.}
\end{figure}

\subsection{Brauer Perturbations and\\ PageRank}
A well-known result on eigenvalue perturbation by Brauer states that for a matrix $A$ with eigenvalues $\lambda_1, \ldots, \lambda_n$, corresponding eigenvector $x$ of $\lambda_i$, and any choice of vector $v \in \CC^n$, $A + xv^*$ has eigenvalues $\lambda_1 + v^*x, \lambda_2, \ldots, \lambda_n$ \cite{PageRank}. 

Thus $\alpha \lambda_1 + (1-\alpha) v^*x$ for $\alpha \in [0,1]$ is a continuously varying eigenvalue present in $Co(A,xv^*)$. Thus we see that there is a 1-component containing $\lambda_1$ and $v^*x$. Further, there is an $(n-1)-$component containing $0$ (as an eigenvalue of $xv^*$ of multiplicity $n-1$) and $\alpha \lambda_{2}, \ldots, \alpha \lambda_n$ for $\alpha \in [0,1]$.

As in standard PageRank, let $S \in M_n$ be a row stochastic matrix, let $e \in \RR^n$ be the all ones vector, and let $v \in \RR^n$ be any probability vector. $S$ has $1$ as an eigenvalue with eigenvector $e$---label the eigenvalues of $S$ as $1, \lambda_2, \ldots, \lambda_n$. The google matrix $G$ is some choice of convex combination $\alpha S + (1-\alpha)ev^T$. Since $e$ is an eigenvalue of $S$, each of these choices of $\alpha$ is a Brauer perturbation of $\alpha S$, and $v^T e = 1$, so the eigenvalues of $\alpha S + (1-\alpha)ev^T$ are $1, \alpha \lambda_2, \ldots, \alpha \lambda_n$.

In particular, the spectra of $Co(S, ev^T)$, which contains all possible choices of $G$, is contained in a $1-$component, and the rest of the eigenvalues $\alpha \lambda_i$ for $i = 2, \ldots, n$ and $\alpha \in [0,1]$ are in an $(n-1)-$component. The pairing structure between eigenvalues is notably clear here, as each eigenvalue of $S$ pairs with $1$ or $0$.

\subsection{Tri-Diagonal Toeplitz}
Suppose $A_1, \ldots, A_k$ are tri-diagonal toeplitz matrices. Then all convex combinations $\sum_i \alpha_i A_i$ are tri-diagonal toeplitz as well, with diagonal entries given by $\sum_i \alpha_i a_i$, where $a_i$ is the diagonal element of $A_i$, and so on. Likewise, if $b_i$ are the super-diagonal elements of $A_i$ and $c_i$ are the sub-diagonal elements of $A_i$, then 
\[\sum_i \alpha_i A_i = \begin{bmatrix}
\sum_i \alpha_i a_i & \sum_i \alpha_i b_i & & & \\
\sum_i \alpha_i c_i & \sum_i \alpha_i a_i & \sum_i \alpha_i b_i &  & \\
 & \ddots & \ddots & & \\
& & \sum_i \alpha_i c_i & \sum_i \alpha_i a_i
\end{bmatrix}\]
and, using the known formula \cite{Toep} for eigenvalues of tri-diagonal toeplitz matrices, has eigenvalues
\[\lambda_k = \sum_i \alpha_i a_i + 2\sqrt{(\sum_i \alpha_i b_i)(\sum_i \alpha_i c_i)} \cos(\frac{\pi k}{n+1})\].

\begin{figure}[hbt!]
    \centering
    \includegraphics[scale=.4]{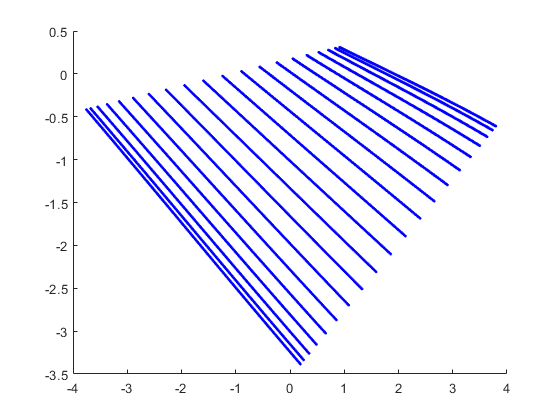}
    \caption{The pairings between two Tri-Diagonal Toeplitz matrices}
\end{figure}

It follows then that these sets are weakly transitive, and if the off-diagonal matrices are never 0 in the convex hull, we further see that the eigen-surface is comprised of $n$ 1-components.

\subsection{Shared Eigenvectors}a

Suppose $A_1, \ldots, A_k$ are a family of matrices in which all eigenvectors are shared. We can take them to be circulant matrices \cite{Circ}, commuting matrices \cite{MatAnal}, diagonal matrices, etc. As we can simultaneously triangularize all matrices in \\$Co(A_1,\dots, A_k)$, it is clear that these eigenvalues are weakly transitive as they are well-defined continuous functions of the convex coefficients.
\begin{figure}[hbt!]
    \centering
    \includegraphics[scale=.4]{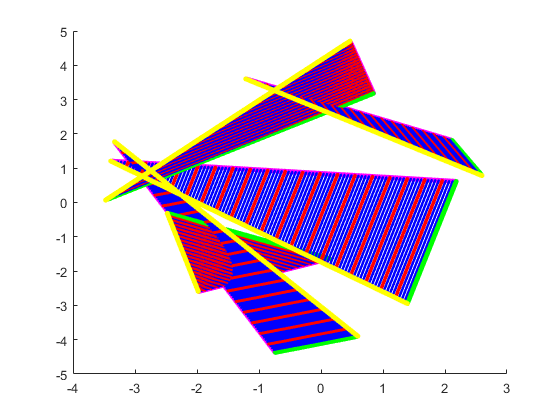}
    \caption{The eigen-surface of the convex hull of three circulant matrices.}
\end{figure}
The criteria then for these sets to be transitive is to verify that these continuous eigenvalue functions never intersect.
\subsection{Dimensional Relationships}

Here we only consider pairings that occur in the line-segments between matrices.

\begin{lemma}
        If $\lambda(T)\rightsquigarrow\lambda(A+B)$, $\lambda(A+B)\rightsquigarrow \lambda(A),\lambda(B)$, then there exist continuous paths $p_A(t),p_B(t)$ in $ES(Co(T,A,B))$, such that $p_A(0) = p_B(0) = \lambda(T)$, $p_A(1) = \lambda(A),p_B(1)=\lambda(B).$
\end{lemma}

\begin{proof}
By hypothesis, $\lambda(T)\rightsquigarrow \lambda(A+B)$, and a simple algebraic computation reveals that pairings are unchanged by positive scaling, so we have that $$\lambda(T)\rightsquigarrow \lambda(1/2(A+B))\in ES(Co(T, A, B)).$$ We also have, by the same way, that $\lambda(1/2(A+B))\rightsquigarrow\lambda(A)$. Thus if we take the function

$$p_A(t):=
\begin{cases}
\lambda(T(2t)),& t\in[0,1/2]\\
\lambda(1/2(A+B)(2t),& t\in[1/2,1]\\
\end{cases}
$$
where $T(x)$ is the line-segment from $T$ to $1/2(A+B)$ and $1/2(A+B)(x)$ is the line-segment from $1/2(A+B)$ to $A$, we get the desired function ($p_B(t)$ is defined similarly).
\end{proof}

\begin{figure}[hbt!]
    \centering
    \includegraphics[scale=.4]{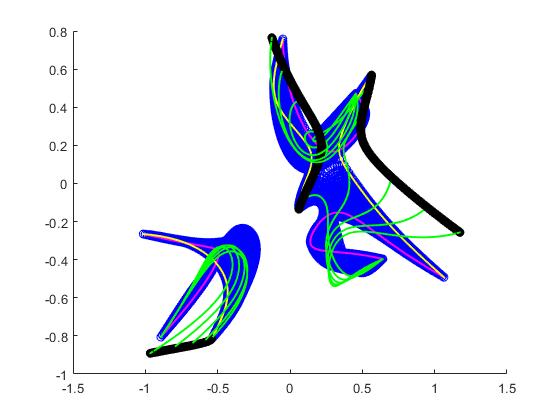}
    \caption{The eigenvalues of $Co(A,B,C)$, overlayed with the pairings by line-segment of $A$ and $A+xB$ as $x$ varies.}
\end{figure}

 We also see that this result generalizes As we can take sums of $k$ matrices $A_1+\dots + A_k$ with $T$, and thus obtain a relationship between eigenvalues of $T$ and $A_1,\dots,A_k$ subject to paths in $Co(A_1,\dots,A_k).$
 
 \section{Open Problems}
 
 \begin{itemize}
     \item Efficient methods to compute eigenvalue pairings and higher multiplicity eigenvalues in the convex hull of arbitrary matrices.
     
     \item Partitioning of the convex hull of arbitrary matrices into finitely many weakly transitive, simply connected, sets.
     
     \item Characterization of the graphs generated by the eigen-surface of the convex hull of arbitrary matrices.
     
     \item Detailed topological characterization of\\ $ES(DS_n)$.

     \item From our results, it follows that if any closed polygonal-path generating the boundary of the convex hull are not transitve, then the set must contain matrices with higher multiplicity eigenvalues. Thus a study of the probability of transitivity in these paths for random matrices, such as those drawn from the Gaussian measure, would reveal important general characterization of the eigen-surface.
 \end{itemize}

\begin{center}
    \textbf{Acknowledgements}
\end{center}
    This project was greatly advanced by the insightful collaboration of fellow REU participants Derek Lim, Eric Jankowski and Amit Harlev, as well as the generous support provided by NSF grant 1757603.

\end{document}